\documentclass[12pt]{amsart}

\usepackage{amssymb,amsmath,amsthm,amscd,mathrsfs,mathtools}
\usepackage{thmtools,thm-restate}
\usepackage{paralist,enumitem,multicol,stackrel,framed,color}
\usepackage{graphicx,pdfpages,imakeidx}
\usepackage{lipsum}
\usepackage[letterpaper,margin=1in]{geometry}
\usepackage{caption}
\usepackage{tikz-cd}

\theoremstyle{plain}

\newtheorem{theorem}{Theorem}[section]

\newtheorem{lemma}[theorem]{Lemma}
\newtheorem{proposition}[theorem]{Proposition}

\theoremstyle{definition}
\newtheorem{definition}[theorem]{Definition}

\theoremstyle{remark}
\newtheorem{remark}[theorem]{Remark}
\newtheorem{algorithm}[theorem]{Algorithm}
\newtheorem{example}[theorem]{Example}

\newtheorem{question}[theorem]{Question}

\numberwithin{equation}{section}

\DeclareMathOperator{\bbC}{\mathbb C}

\DeclareMathOperator{\bbD}{\mathbb D}

\DeclareMathOperator{\bbH}{\mathbb H}

\DeclareMathOperator{\bbR}{\mathbb R}

\newcommand{\vcentering}[1]{\raisebox{-0.5\height}{#1}}

\newcommand{\ds}{\displaystyle}

\title{Closed Cap Condition under the Cap Construction Algorithm}
\author{Mercedes Sandu}
\address{Department of Mathematics, Northwestern University, Evanston, IL 60208}
\email{MercedesSandu2024@u.northwestern.edu}

\author{Shuyi Weng}
\address{Department of Mathematics, Northwestern University, Evanston, IL 60208}
\email{ShuyiWeng2015@u.northwestern.edu}

\author{Jade Zhang}
\address{Department of Mathematics, Northwestern University, Evanston, IL 60208}
\email{JadeZhang2024@u.northwestern.edu}

\begin{document}

\maketitle

\begin{abstract}
Every polygon $P$ can be companioned by a cap polygon $\hat P$ such that $P$ and~$\hat P$ serve as two parts of the boundary surface of a polyhedron $V$.
Pairs of vertices on $P$ and~$\hat P$ are identified successively to become vertices of $V$.
In this paper, we study the cap construction that asserts equal angular defects at these pairings.
We exhibit a linear relation that arises from the cap construction algorithm, which in turn demonstrates an abundance of polygons that satisfy the \emph{closed cap condition}, that is, those that can successfully undergo the cap construction process.
\end{abstract}

\section{Introduction}
\label{sec:introduction}

A planar shape is a subset $P \subseteq \mathbb R^2$ that is compact, connected, contains at least two points, and has a  connected complement.
There are many ways to construct ``caps'' of $P$ that glue to $P$ along the boundary to form topological spheres.
We are interested in the following question: is it possible to have the topological sphere adopt a specified distribution of curvature along the common boundary of $P$ and its cap?
DeMarco and Lindsey~\cite{DeMarco2017} showed that such construction is always possible when the curvature is the harmonic measure on the boundary of $(\bbR^2 \backslash P, \infty)$.
The second author studied flatness of three-dimensional realizations of the topological spheres in his PhD thesis~\cite{Weng2020thesis}.
In this paper, we investigate the case when~$P$ is a polygon and when the curvature distribution is uniform on the vertices of~$P$.
On the resulting polyhedral topological sphere, curvature is concentrated atomically on vertices, manifested as angular defects at vertices.
We adopt the following definition of a \emph{polygon} from T\'oth, Goodman, and O'Rourke~\cite[Chapter 30]{Toth2017handbook}.

\begin{definition}[Polygon]
\label{def:polygon}
An \textbf{$n$-sided polygon} is a closed region of the plane enclosed by a simple cycle of $n$ straight line segments.
The \textbf{vertices} of the polygon are the endpoints of these line segments.
\end{definition}

A cap $\hat P$ of a polygon $P$ is expected to glue onto $P$ along their boundaries to form a topological sphere.
So, they must have the same circumference length.
Additionally, if the curvature is distributed atomically on the vertices of $P$, the topological sphere must be a polyhedron.
Therefore, $\hat P$ is expected to be an $n$-sided polygon with the same side lengths as $P$.
If $v_1, v_2, \dots, v_n$ is an enumeration of vertices of $P$ in counterclockwise order, and if $\hat P$ is a cap of $P$, we expect to have a corresponding enumeration $\hat{v}_1, \hat{v}_2, \cdots, \hat{v}_n$ of vertices of $\hat P$ listed in clockwise order, such that the length of the segment $[v_k, v_{k+1}]$ equals the length of the segment $[\hat v_k, \hat v_{k+1}]$ (with $v_{n+1}=v_1$ and $\hat v_{n+1} = \hat v_1$).
Thus, every edge of $P$ can be glued to a corresponding edge of $\hat{P}$ with equal length.

There are many caps that can be constructed for a given polygon $P$. To start, a reflected image of $P$ can be placed on top of $P$, gluing only along the boundary to form a doubly-covered planar region, also called a degenerate polyhedron.
If a convex polygon $P$ is folded along one diagonal by some non-trivial amount, consider the boundary surface of the convex hull: part of its surface can be attributed to $P$ folded up, while the development of the complement will be the cap $\hat P$.
Figure~\ref{fig:capExample} shows an example of such construction.
We see that there is a lot of freedom in constructing caps of polygons.
We want to impose a further restriction on the angular defects at the vertices of the resulting polyhedral topological sphere---that is, each vertex has the same angular defect.
Now the question becomes, is this restricted construction always possible? What are some conditions that allow for it?
We will address these two questions in this paper.

\begin{figure}[t]
    \centering
    \includegraphics[width=\textwidth]{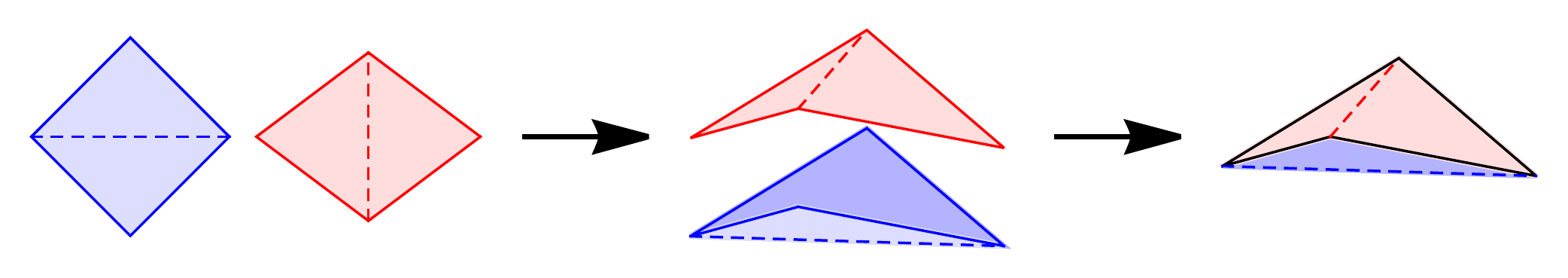}
    \caption{A possible cap of a square is glued to the square along their boundary to form the boundary surface of a tetrahedron.}
    \label{fig:capExample}
\end{figure}

\subsection{The cap construction algorithm}
\label{subsec:algorithm}

We describe the {cap construction algorithm} as the discrete version of \emph{perimeter gluing}~\cite[Theorem~1.2]{DeMarco2017}.
First, begin with an $n$-sided polygon~$P$, with vertices $v_1, v_2, \cdots, v_n$ listed in counterclockwise order.
Because $P$ is uniquely determined by its vertices, the side length $\ell_k$, that is, the length of the edge connecting vertices $v_k$ and $v_{k+1}$ (with $v_{n+1}=v_1$), can be computed, as can the internal angle $\theta_k$ at vertex $v_k$.
Gluing the edge $[\hat v_{k-1}, \hat v_k]$ of $\hat P$ to the corresponding edge $[v_{k-1}, v_k]$ of $P$ results in a configuration shown in Figure~\ref{fig:angularDefect}.

\begin{figure}[b]
 \centering
 \captionsetup{justification=centering}
 \includegraphics{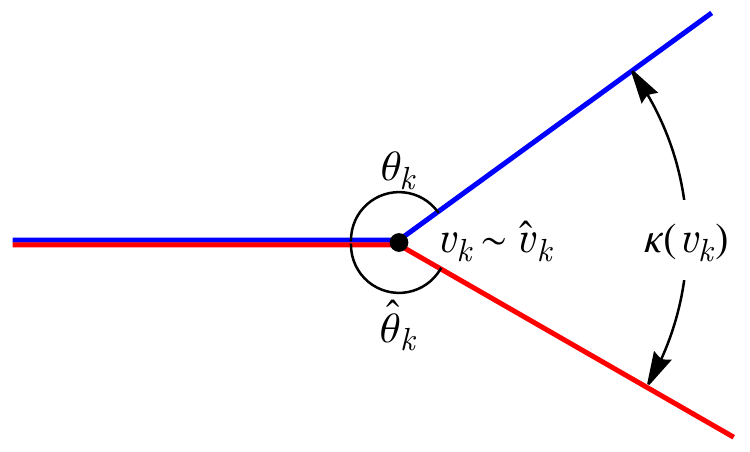}
 \caption{Angular defect at $v_k$.}
 \label{fig:angularDefect}
\end{figure}

If the resulting polyhedron is flat at $v_k$, the internal angles must add up to $2\pi$, so that
$\theta_k + \hat{\theta}_k = 2\pi$.
Otherwise, there is an angular defect, which is the discrete curvature at the glued $v_k$. The angular defect is calculated by:
\[ \kappa(v_k) = 2\pi - \theta_k - \hat{\theta}_k. \]
Descartes' theorem on the total defect of a polyhedron~\cite[p.~61,~Proposition~1]{Federico1982Descartes} claims that the sum of angular defects on all vertices of a polyhedron is $4\pi$.
In order to achieve uniform distribution of curvature on vertices, we need
\[ \kappa(v_k) = \frac{4\pi}{n}. \]

We can summarize the cap construction algorithm as follows:
\begin{algorithm}[Cap construction algorithm]
\label{alg:capConstruction}
Let $P$ be an $n$-sided polygon with $v_1, v_2, \dots, v_n$ an enumeration of its vertices in counterclockwise orientation.
\begin{enumerate}[nolistsep,itemsep=0pt]
    \item Assume $v_1 = (0,0)$ and $v_2 = (1,0)$ up to rotation, translation, and scaling.
    \item Assign angular defect at each vertex with $\kappa(v_k) = 4\pi / n$.
    \item Let $\hat v_1 = v_1$ and $\hat v_2 = v_2$.
\end{enumerate}
Now, starting with $k=2$,
\begin{enumerate}[nolistsep,itemsep=0pt]
\setcounter{enumi}{3}
    \item At the vertex $\hat v_k$, calculate the internal angle $\hat{\theta}_k = 2\pi - \theta_k - \kappa(v_k)$.
    If $\hat{\theta}_k > 0$, proceed to the next step. Otherwise, we cannot construct the cap.
    \item The (counterclockwise) turning angle $\alpha_k$ at the vertex $v_k$ is supplementary to the internal angle $\theta_k$. Thus $\alpha_k = \pi - \theta_k$.
    \item
    The (clockwise) turning angle $\beta_{k}$ at the vertex $\hat v_k$ is supplementary to the internal angle $\hat \theta_k$.
    Thus $\beta_{k} = \pi - \hat{\theta}_k$.
    \item Draw the edge $[\hat v_{k}, \hat v_{k+1}]$ of $\hat{P}$ in the direction of turning angle $\beta_k$ from the previous edge $[\hat v_{k-1}, \hat v_{k}]$, with the same length as $[v_k, v_{k+1}]$ in $P$.
    \item Repeat steps (4) through (7) for all $k \leq n$, and when $k=n$, take $v_{n+1} = v_1$.
    When the process is completed, we obtain a polygonal cap curve defined by vertices $\hat v_1, \dots, \hat v_{n+1}$.
\end{enumerate}
\end{algorithm}

\emph{If this algorithm successfully produces the boundary curve of a polygonal cap $\hat P$}, then it endows a metric whose curvature is supported on the set of vertices of~$P$ and~$\hat P$.
Alexandrov's uniqueness theorem~\cite[p.~100,~Theorem~1*]{Alexandrov2006Convex} then applies to guarantee a unique 3-dimensional realizations of gluing the polygon $P$ with its cap $\hat P$ along the boundary.

\begin{theorem}[Alexandrov]
\label{thm:AlexandrovUniqueness}
Every development homeomorphic to the sphere and having the sum of angles at most $2\pi$ at each vertex defines a closed convex polyhedron (possibly degenerate as a doubly-covered convex polygon) by gluing. Furthermore, the polyhedron is unique up to rigid motions.
\end{theorem}

\begin{example}
\label{ex:antiprism}
Let $P$ be a square, while treating it as an octagon with an additional vertex at the center of each side (as shown in Figure~\ref{fig:antiprism}).
The cap construction algorithm applied on $P$ gives a square of the same size.
However, sides of the squares are identified as offset by half their side length, giving the boundary surface of a square antiprism.
See Figure~\ref{fig:antiprism}.
\begin{figure}[t]
    \centering
    \includegraphics[width=\textwidth]{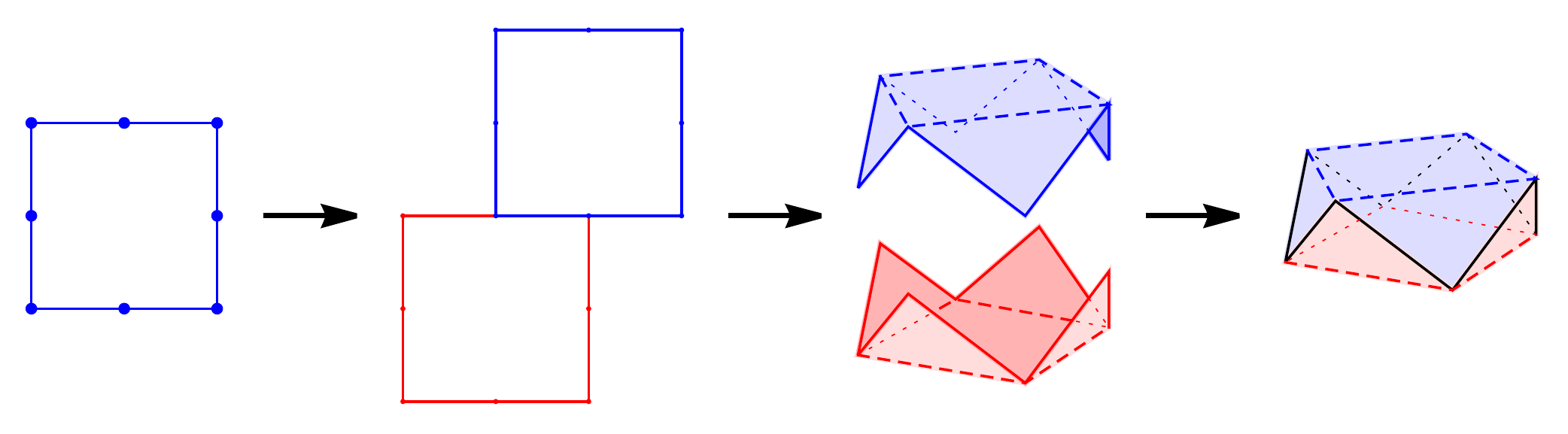}
    \caption{The cap construction algorithm applied on a square octagon, as described in Example~\ref{ex:antiprism}}
    \label{fig:antiprism}
\end{figure}
\end{example}

\begin{remark}
\label{rem:optionalStep1}
The first step of the cap construction algorithm is for computation convenience and is not strictly necessary.
The absolute location of $P$ within $\mathbb R^2$ does not affect the construction of the 3D realization once a polygonal cap is obtained.
Further, similar polygons produce similar polygonal cap curves under the cap construction algorithm without the first step, which we will prove in Lemma~\ref{lem:similarityEquivalence}.
\end{remark}

\subsection{The closed cap condition}
\label{subsec:closedCapCondition}

The cap construction algorithm gives a polygonal cap curve with $n+1$ vertices, which may not be a simple closed curve to begin with, let alone the boundary curve of a polygon.
Figures~\ref{fig:closed} and~\ref{fig:notClosed} show two step-by-step cap construction diagrams for one polygon that results in a closed polygonal cap curve and another polygon that does not, respectively.
Figure~\ref{fig:capConstructionFailures} shows the construction of a closed but not simple polygonal curve.

\begin{figure}[bp]
 \centering
 \captionsetup{justification=centering}
 \includegraphics{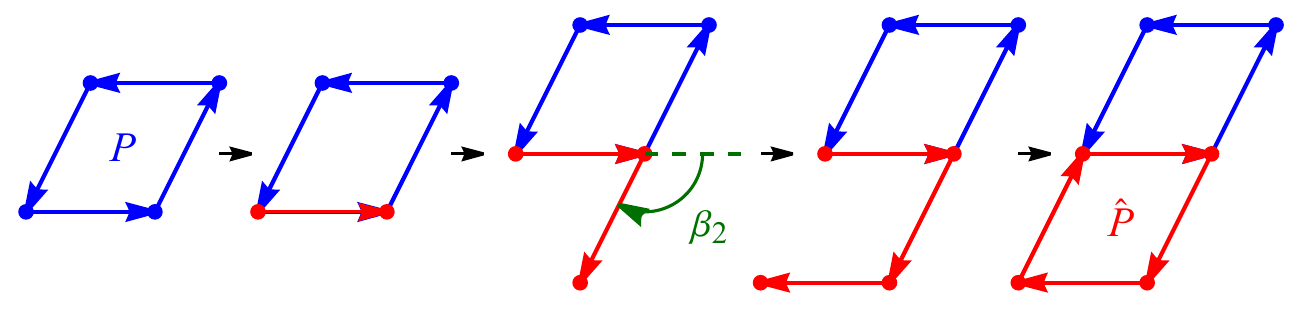}
 \caption{Cap construction where $P$ results in $\hat{P}$ that is closed.}
 \label{fig:closed}
\end{figure}

\begin{figure}[bp]
 \centering
 \captionsetup{justification=centering}
 \includegraphics{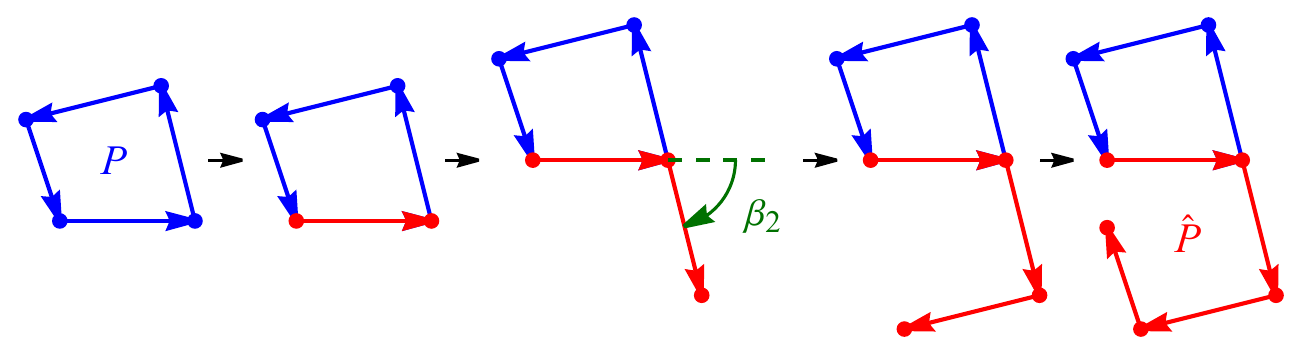}
 \caption{Cap construction where $P$ results in $\hat{P}$ that is not closed.}
 \label{fig:notClosed}
\end{figure}

\begin{figure}[tp]
 \centering
 \captionsetup{justification=centering}
 \includegraphics[angle=90]{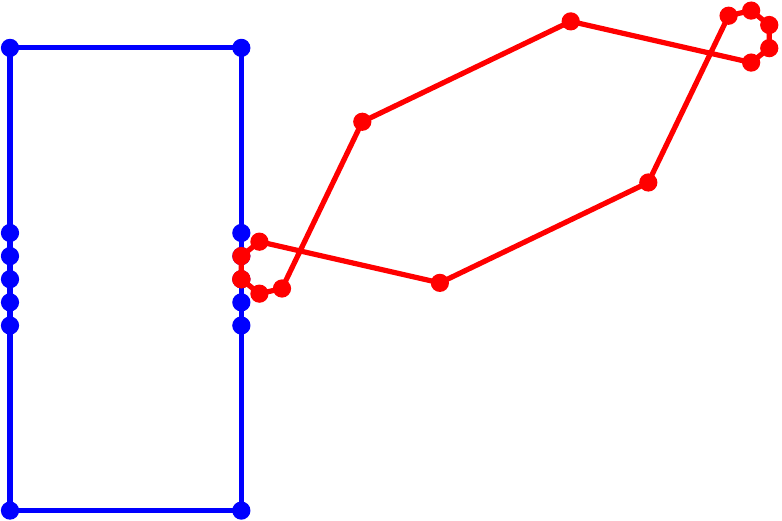}
 \caption{Failure of the cap construction algorithm.}
 \label{fig:capConstructionFailures}
\end{figure}

In this paper, we will focus on the open-curve type failure, and will primarily investigate the condition for which the cap construction algorithm gives rise to a closed polygonal curve.
The question of whether the resulting polygonal cap curve self-intersects is not within the scope of this paper, and we leave it as an open problem in Question~\ref{q:self-intersect}.

\begin{definition}[Closed cap condition]
An $n$-sided polygon $P$, with an enumeration of its vertices $v_1, v_2, \dots, v_n$ in counterclockwise order, is said to satisfy the \textbf{closed cap condition} if the resulting vertices $\hat v_{n+1} = \hat v_{1}$ under the cap construction algorithm as described in Algorithm~\ref{alg:capConstruction}.
\end{definition}

\subsection{Main results}

As we have seen in Figure~\ref{fig:notClosed}, not all polygons satisfy the closed cap condition.
In this paper, we prove that the gap $\hat v_{n+1} - \hat v_{1}$ between endpoints of the polygonal cap curve depends linearly on all vertices $v_k$ of the polygon $P$.

Besides the initial assumption that $v_1 = (0,0)$ and $v_2 = (1,0)$ (cf. Remark~\ref{rem:optionalStep1}), the cap construction algorithm (Algorithm~\ref{alg:capConstruction}) was not particularly specified in $\mathbb R^2$. We can use the standard identification of $\mathbb R^2$ and $\mathbb C$ to state our main results.

\begin{restatable}{theorem}{affineDependence}
\label{thm:affineDependence}
Let $n \geq 3$ be an integer, and let $P$ be an $n$-sided polygon with vertices $v_1, \dots, v_n$, enumerated in counterclockwise order. Then the gap $\hat v_{n+1} - \hat v_1$ of the polygonal cap curve under the cap construction algorithm is given by
\begin{equation}
\label{eq:linearRelation}
\hat v_{n+1} - \hat v_1 = \sum_{k=1}^n (1-\omega)\omega^{k-2}v_k,
\end{equation}
where $\omega = e^{-4\pi i/n}$.
\end{restatable}

A consequence of Theorem~\ref{thm:affineDependence} is that the closed cap condition is equivalent to a linear relation involving all vertices of an $n$-sided polygon.

\begin{restatable}{theorem}{linearRelation}
\label{thm:linearRelation}
For any integer $n \geq 3$, the space of $n$-sided polygons satisfying the closed cap condition is locally $(n-1)$-complex dimensional. In particular, an $n$-sided polygon with vertices $v_1, v_2, \dots, v_n \in \bbC$ enumerated in counterclockwise order satisfies the closed cap condition if and only if
\[
\sum_{k=1}^n \omega^k v_k = 0,
\]
where $\omega = e^{-4\pi i/n}$.
\end{restatable}


\subsection{Organization of the paper}
In Section~\ref{sec:Han18}, we summarize previous work by Han~\cite{Han2018poster} which classified all triangles and quadrilaterals that satisfy the closed cap condition.
Specifically, they are equilateral triangles and parallelograms, respectively.
In Sections~\ref{sec:case} and~\ref{sec:affine}, we first present a case study on pentagons, where we fix four vertices at the vertices of the unit square, and allow the fifth vertex to move freely in $\bbR^2$.
We prove that there is a unique position for the fifth vertex such that the resulting pentagon satisfies the closed cap condition.
We further generalize the method used in the proof to demonstrate an affine dependence of the closed cap condition on all vertices.
We prove Theorems~\ref{thm:affineDependence} and~\ref{thm:linearRelation}, which show that polygons satisfying the closed cap condition are abundant.
We conclude in Section~\ref{sec:future} with a list of some open questions pertaining to the cap construction process.

\subsection{Other comments and acknowledgement}

The work of Demaine and O'Rourke~\cite{Demaine2007Folding} motivated us to classify $n$-sided polygons (that result in congruent convex polyhedrons) through creases from Alexandrov's realization.
The main result of this paper is an answer to the very first question of this study---which polygons satisfy the closed cap condition to begin with?
All figures in this paper were created with Mathematica code improved from Mathematica notebooks available from~\cite{Han2018poster} and \cite{Weng2020thesis}.
A lot of paper-folding went into realizing the convex polyhedrons arising from successful cap constructions.
Alexandrov's uniqueness theorem~\cite{Alexandrov2006Convex} guarantees that these convex realizations are indeed good models for the topological sphere with prescribed curvature.
These polyhedral realizations have been deeply involved in several studies under various contexts.
For example, DeMarco and Lindsey used the polygonal realization to approximate \emph{harmonic caps}~\cite{DeMarco2017}; Han classified all polyhedral realizations that arise from triangles and quadrilaterals~\cite{Han2018poster}; Weng investigated the condition for degenerate convex realizations under harmonic cap construction~\cite{Weng2020thesis}.

We would like to thank Laura DeMarco for introducing the harmonic cap problem to the second author and for her suggestion on adapting certain aspects of this problem into an undergraduate research project.
We also thank Aaron Peterson and Santiago Ca\~nez for their helpful advice.
Finally, we thank the anonymous referee for many thoughtful, useful editing suggestions.

This research was supported by a grant from the Undergraduate Research Assistant Program which is administered by the Office of Undergraduate Research at Northwestern University.

\section{Known Results on Triangles and Quadrilaterals}
\label{sec:Han18}

A previous study by Han~\cite{Han2018poster} focused on triangles and quadrilaterals. The study successfully classified all triangles and quadrilaterals that satisfy the closed cap condition.
We summarize their results here together with their proofs.

The classification for triangles is relatively straightforward.

\begin{proposition}
If $P$ is a triangle that satisfies the closed cap condition, then $P$ must be an equilateral triangle.
\end{proposition}

\begin{proof}
Under the folding rule, we only have the three vertices of $P$  for the convex realization.
Thus the associated convex body is planar, namely a double-cover of the triangle $P$.
At each vertex with internal angle $\theta$, the angular defect is $2\pi-2\theta$.
For equal distribution of angular defect, we need $2\pi-2\theta = 4\pi/3$ at each vertex.
Thus all three vertices have internal angle $\theta = \pi/3$, giving an equilateral triangle.
\end{proof}

The classification for quadrilaterals is more interesting---we have more than a single similarity class of quadrilaterals that satisfy the closed cap condition.

\begin{proposition}
\label{prop:parallelograms}
If $P$ is a quadrilateral that satisfies the closed cap condition, then $P$ must be a parallelogram.
\end{proposition}

\begin{proof}
The supposed angular defect at each vertex is $4\pi/4=\pi$.
Thus $\alpha_k + \beta_k = \pi$ for all $k$.
Without loss of generality, assume that the first two vertices are at $0,1 \in \mathbb C$, and the next two at $a,b \in \mathbb C$, traced in counterclockwise order.
This is justified because every line segment differs from $[0,1]$ by a conformal affine map, or equivalently, a complex linear polynomial map.
With the condition that $\theta_k \leq \pi$, we can further assume that $a,b \in \overline{\mathbb H} = \{ z \in \bbC \mid \mathrm{Im}(z) \geq 0 \}$.
The closed cap condition can therefore be formulated algebraically by
\begin{align*}
(1-0) + \exp[-i(\alpha_2+\beta_2)](a-1) + \exp[-i(\alpha_2 + \alpha_3 + \beta_2 + \beta_3)](b-a) & \\
+ \exp[-i(\alpha_2 + \alpha_3 + \alpha_4 + \beta_2 + \beta_3 + \beta_4)](0-b) &= 0.
\end{align*}
Because $\alpha_k + \beta_k = \pi$, we get $(1-0) - (a-1) + (b-a) - (0-b) = 0$.
Thus $a-b=1$.
Therefore, $0,1,a,b \in \mathbb C$ form a parallelogram with vertices traced in counterclockwise direction.
\end{proof}

In addition to the classification problem, Han's study also investigated the resulting convex polyhedrons from parallelograms.
Note that the third vertex $v_3 = a$ uniquely determines a parallelogram in $\overline{\mathbb H}$.
Thus every parallelogram with $v_1 = 0$ and $v_2 = 1$ can be parameterized by $\tau \in \mathbb H$, serving as $v_3$.
We will refer to this parallelogram as $Q(\tau)$.
\begin{lemma}[\cite{Han2018poster}]
Within the one-parameter family $\{Q(\tau)\}_{\tau \in \mathbb H}$ of parallelograms,
\begin{enumerate}[label=\textup{(\arabic*)}]
\item $Q(\tau)$ and $Q(\tau+1)$ yield congruent convex realizations;
\item $Q(\tau)$ and $Q(-1/\tau)$ are similar, thus they yield similar convex realizations.
\end{enumerate}
\end{lemma}

\begin{remark}
We can observe that $\tau \mapsto \tau+1$ and $\tau \mapsto -1/\tau$ generates the modular group $\mathrm{PSL}(2,\mathbb Z)$.
As a result of the claim above, if $a,b,c,d \in \mathbb Z$ with $ad-bc=1$, then
\[
Q \Big( \frac{a\tau + b}{c\tau + d} \Big)
\]
yield a convex realization that is similar with that of $Q(\tau)$.
Within the fundamental domain $D = \{ \tau \in \mathbb H \mid |\tau|>1, |\mathrm{Re}(\tau)|<1/2 \}$ of the modular group, each $\tau \in D$ develops a different parallelogram, which then folds into a different convex realization.
\end{remark}

\section{Case Study: Uniqueness of Fifth Point on a Pentagon}
\label{sec:case}

In this section, we present a preliminary case study to show that given four specific vertices, there exists a unique fifth vertex that forms a pentagon satisfying the closed cap condition.
We consider pentagons in $\bbC$ with four of their vertices fixed at $0$, $1$, $1+i$, and $i$.

\begin{proposition}
\label{prop:affinePentagonUnitSquare}
Consider the cap construction algorithm on a pentagon $P$ with vertices
\[
v_1 = 0, \ v_2 = 1, \ v_3 = 1+i, \ v_4 = i, \text{ and } \ v_5 \in \bbC
\]
in counterclockwise order.
The function ${f}: \mathbb{C} \to \mathbb{C}$ that takes the fifth point $v_5$ of the pentagon as its input and returns the endpoint of the polygonal cap curve as its output is an affine function of the form $f(v_5) = av_5+b$ for some $a, b \in \bbC$.
Furthermore, there exists a unique point $v_5 \in \bbC$ such that the pentagon $P$ satisfies the closed cap condition.
\end{proposition}

\begin{proof}
For a pentagon $P$ with four of its vertices $0$, $1$, $1+i$, and $i$ in $\bbC$, the cap construction algorithm gives
\[
f(v_5) = 1 + \exp\left[-\frac{3\pi}{10}i\right] + \exp\left[-\frac{3\pi}{5}i\right] + (v_5-i)\exp\left[-\frac{12\pi}{5}i\right]+ (-v_5)\exp\left[-\frac{16\pi}{5}i\right].
\]
This function realizes the cap construction algorithm to calculate the endpoint $\hat{v}_6$ of the polygonal cap curve.
The function is a sum of five terms, each of which corresponds to an edge $\hat{s}_k = \hat{v}_{k+1} - \hat v_k$ of the polygonal cap curve.
Figure~\ref{fig:function} shows the geometry of the function~$f$ as described above.

\begin{figure}[tp]
 \centering
 \captionsetup{justification=centering}
 \includegraphics[height=5in]{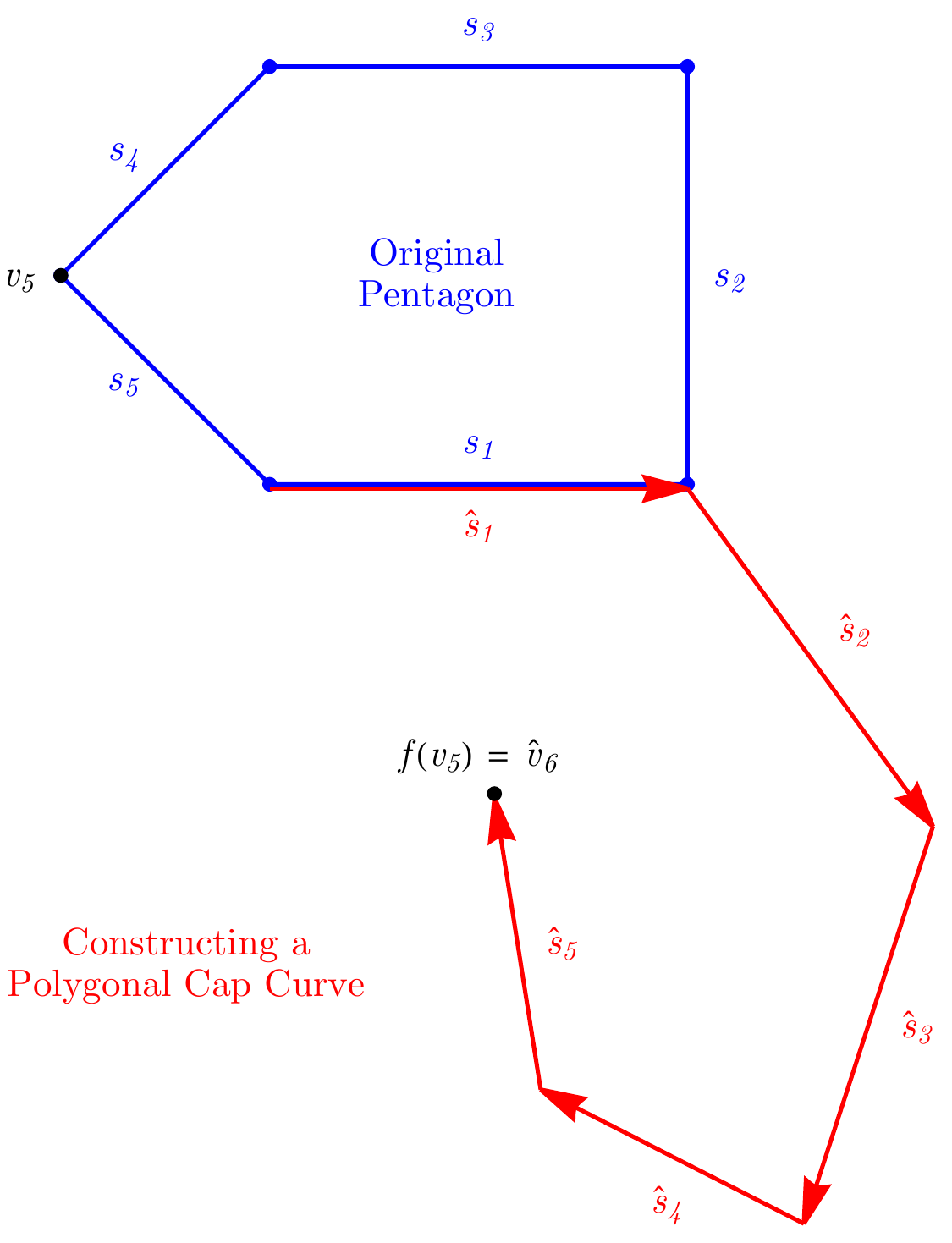}
 \caption{Geometry of the function $f$.}
 \label{fig:function}
\end{figure}

Rearranging the formula of $f$, we see that
\[
f(v_5) = \left( e^{-12\pi i/5} - e^{-16\pi i/5} \right) v_5 + \left( 1 + e^{-3\pi i/10} + e^{-3\pi i/5} - ie^{-12\pi i/5} \right).
\]
Observe that $a = e^{-12\pi i/5} - e^{-16\pi i/5} \neq 0$.
Hence, $f$ is invertible, and there is a unique $v_5 \in \bbC$ that satisfies $f(v_5) = 0 = \hat v_1$.
\end{proof}

\section{An Affine Dependence on Polygons}
\label{sec:affine}

In the previous case study, we observed that, fixing four specific vertices, the function ${f}: \bbC \to \bbC$ that takes the fifth point of a pentagon as its input and returns the endpoint of the cap curve as its output is an affine transformation.
In this section, we will generalize this result to all vertices of an arbitrary $n$-sided polygons.
Note that we will drop the assumptions $v_1 = 0$ and $v_2 = 1$ in this section (cf. Remark~\ref{rem:optionalStep1}).

We begin with an observation from Figures~\ref{fig:closed} and~\ref{fig:notClosed} that for an arbitrary quadrilateral~$P$, the $k$-th edge of the polygonal cap curve $\hat P$ is parallel to the $k$-th edge of the original quadrilateral $P$, for integers $1 \leq k \leq 4$, regardless of whether the closed cap condition is met or not.
This observation can be generalized into the following proposition.

\begin{proposition}
\label{prop:parallelEdges}
Let $n \geq 3$ be an integer.
Consider the cap construction algorithm on an arbitrary $n$-sided polygon $P$.
The $k$-th edge $\hat s_k = \hat v_{k+1} - \hat v_k$ of the polygonal cap curve $\hat P$ is related to the $k$-th edge $s_k = v_{k+1} - v_k$ of the original polygon $P$ by
\[
\hat s_k = \omega^{k-1} s_k,
\]
where $\omega = e^{-4\pi i/n}$.
\end{proposition}

\begin{proof}
Let $P$ be a polygon with vertices $v_1, \dots, v_n \in \bbC$, enumerated in counterclockwise order.
Let $s_k = v_{k+1}-v_k$ be the $k$-th edge of $P$ (taking $v_{n+1} = v_1$).
By the cap construction algorithm, the angular defect at each vertex is $\kappa = \alpha_k + \beta_k = 4\pi/n$.
Thus
\begin{equation}
	\label{eq:companionSideRelationComplex}
	\hat s_k = e^{-i(\alpha_2+\beta_2)}e^{-i(\alpha_3+\beta_3)} \cdots e^{-i(\alpha_k+\beta_k)} s_k = e^{-4\pi(k-1)i/n}s_k = \omega^{k-1} s_k. \qedhere
\end{equation}
\end{proof}

\begin{remark}
In the quadrilateral case ($n=4$), all pairs of corresponding edges $(s_k, \hat s_k)$ are indeed parallel. With $\omega = e^{4\pi i /4} = -1$, we can observe that $\hat s_1 = s_1$ and $\hat s_3 = (-1)^2 s_3 = s_3$, while $\hat s_2 = (-1)^1 s_2 = -s_2$ and $\hat s_4 = (-1)^3 s_4 = -s_4$.
\end{remark}

%

\begin{figure}[b]
	\centering
	\captionsetup{justification=centering}
	\vcentering{\includegraphics[width=3.2in]{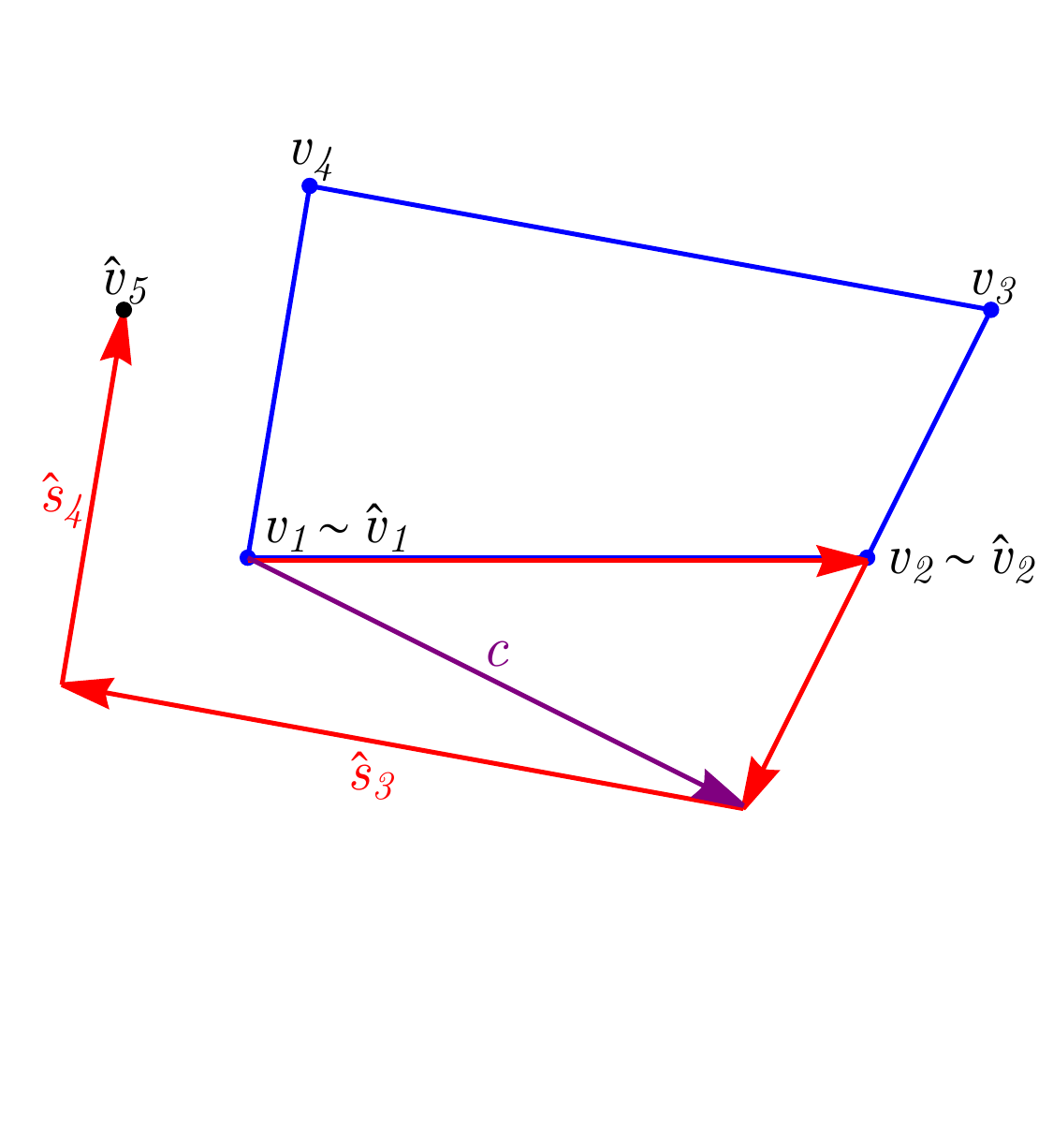}}
	\hfill
	\vcentering{\includegraphics[width=3.2in]{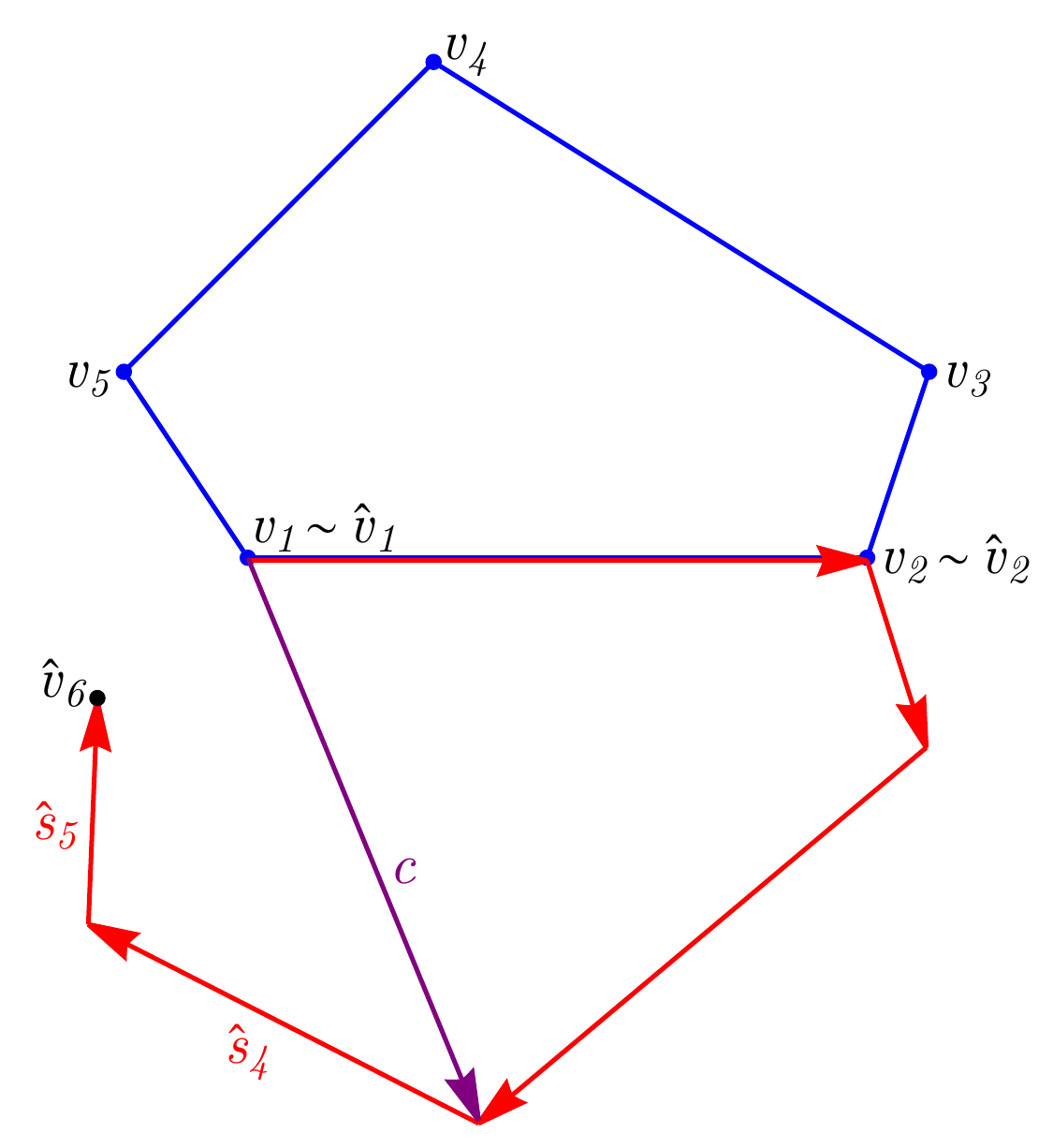}}
	\caption{Arbitrary quadrilateral and pentagon with their respective polygonal cap curves.}
	\label{fig:affinePolygons}
\end{figure}

%
%
%

In the case study of Section~\ref{sec:case},  Proposition~\ref{prop:affinePentagonUnitSquare} says that the endpoint $\hat v_6$ of the polygonal cap curve is affine dependent on the choice of the final vertex $v_5$.
This conclusion can also be generalized to arbitrary $n$-sided polygons.

\begin{proposition}
\label{prop:affineNGon}
Let $n \geq 3$ be an integer. Consider the cap construction algorithm on an arbitrary $n$-sided polygon $P$.
The endpoint $\hat v_{n+1}$ of the polygonal cap curve is affine dependent on the $n$-th vertex $v_n$ of $P$.
\end{proposition}

\begin{proof}
Let $P$ be an $n$-sided polygon with vertices $v_1, \dots, v_n \in \bbC$ enumerated in counterclockwise order, and let $s_k = v_{k+1}-v_k$ be its sides.
Let $c = \hat v_{n-1} - \hat v_1$ (See Figure~\ref{fig:affinePolygons} for quadrilateral and pentagonal cases).
Then
\[
\hat v_{n+1} = c + \hat s_{n-1} + \hat s_n.
\]
With equation~\ref{eq:companionSideRelationComplex}, we have $\hat s_k = \omega^{k-1}s_k$, where $\omega = e^{-4\pi i/n}$.
Thus
\begin{align*}
	\hat v_{n+1} &= c + \omega^{n-2} s_{n-1} + \omega^{n-1} s_n \\
	&= c + \omega^{n-2} (v_n - v_{n-1}) + \omega^{n-1} (v_1 - v_n) \\
	&= (c + \omega^{n-1} v_1 - \omega^{n-2} v_{n-1}) + (\omega^{n-2} - \omega^{n-1}) v_n.
\end{align*}
It follows that $\hat v_{n+1}$ is affine dependent on $v_n$, while all other vertices are held constant.
\end{proof}

Finally, affine dependence can be generalized to any vertex, rather than only the final vertex, of an arbitrary $n$-sided polygon, which gives a proof of Theorem~\ref{thm:affineDependence}.


\affineDependence*

\begin{proof}
Let $P$ be an $n$-sided polygon with vertices $v_1, \dots, v_n \in \bbC$ enumerated in counterclockwise order, and let $s_k = v_{k+1}-v_k$ be its sides.
By equation~\ref{eq:companionSideRelationComplex}, we have $\hat s_k = \omega^{k-1}s_k$, where $\omega = e^{-4\pi i/n}$.
By the cap construction algorithm,
\begin{align*}
\hat v_{n+1} &= \hat v_1 + \hat s_1 + \hat s_2 + \cdots + \hat s_{n-1} + \hat s_n \\
&= v_1 + \omega^0s_1 + \omega^1s_2 + \cdots + \omega^{n-2}s_{n-1} + \omega^{n-1}s_n \\
&= v_1 + \omega^0(v_2 - v_1) + \omega^1(v_3 - v_2) + \cdots + \omega^{n-2}(v_n - v_{n-1}) + \omega^{n-1}(v_1 - v_n) \\
&= \omega^{n-1}v_1 + (\omega^0 - \omega^1)v_2 + (\omega^1 - \omega^2)v_3 + \cdots + (\omega^{n-2} - \omega^{n-1})v_n.
\end{align*}
Thus
\[
\hat v_{n+1} - \hat v_1 = \sum_{k=1}^n (\omega^{k-2} - \omega^{k-1})v_k = \sum_{k=1}^n (1 - \omega)\omega^{k-2}v_k. \qedhere
\]
\end{proof}

Equation~\ref{eq:linearRelation} is a linear relation between vertices $v_k$ and the gap $\hat v_{n+1} - \hat v_1$ between the endpoints of the polygonal cap curve, which also shows that the endpoint $\hat v_{n+1}$ is affine dependent on each $v_k$.
The closed cap condition is precisely $\hat v_{n+1} - \hat v_1 = 0$.
Theorem~\ref{thm:linearRelation} therefore follows readily.

\linearRelation*

\begin{proof}
Observe that
\[
\sum_{k=1}^n \omega^k v_k = (1-\omega)^{-1}\omega^2 \sum_{k=1}^n (1 - \omega)\omega^{k-2}v_k = (1-\omega)^{-1}\omega^2 (\hat v_{n+1} - \hat v_1) = 0.
\]
Thus $\hat v_{n+1} - \hat v_1 = 0$ if and only if $\displaystyle \sum_{k=1}^n \omega^k v_k = 0$.
\end{proof}

\begin{remark}
An immediate consequence from Theorem~\ref{thm:linearRelation} is that the space of $n$-sided polygons satisfying the closed cap condition is locally complex $(n-1)$-dimensional.
With additional assumptions $v_1 = 0$ and $v_2 = 1$ back in place, we can recover the one-parameter family $\{Q_\tau\}_{\tau \in \bbH}$ of parallelograms, which agrees with Han's results presented in Section~\ref{sec:Han18}, particularly Proposition~\ref{prop:parallelograms}.
\end{remark}

We finish our exposition with a justification for a claim in Remark~\ref{rem:optionalStep1}.

\begin{lemma}
\label{lem:similarityEquivalence}
If $P$ and $Q$ are similar polygons, then their respective polygonal cap curves $\hat P$ and $\hat Q$ under the cap construction algorithm are similar.
\end{lemma}

\begin{proof}
Let $P$ be an $n$-sided polygon with vertices $v_1, \dots, v_n \in \bbC$ enumerated in counterclockwise order.
As a similar polygon, $Q$ must have an enumeration of vertices $w_1, \dots, w_n \in \bbC$ such that $w_k = av_k+b$, where $a,b \in \bbC$ and $a \neq 0$.
By the cap construction algorithm, we have
\begin{align*}
\hat w_{k+1}
&= w_1 + \omega^0(w_2 - w_1) + \omega^1(w_3 - w_2) + \cdots + \omega^{k-1}(w_{k+1} - w_{k}) \\
&= (av_1+b) + a\omega^0(v_2 - v_1) + a\omega^1(v_3 - v_2) + \cdots + a\omega^{k-1}(v_{k+1} - v_{k}) \\
&= a [ v_1 + \omega^0(v_2 - v_1) + \omega^1(v_3 - v_2) + \cdots + \omega^{k-1}(v_{k+1} - v_{k}) ] + b \\
&= a \hat v_{k+1} + b
\end{align*}
for all $2 \leq k \leq n$, where we take $w_{k+1} = w_1$ and $v_{k+1} = v_1$.
Thus the polygonal cap curves $\hat P$ and $\hat Q$ are similar.
\end{proof}

\section{Open Questions and Future Work}
\label{sec:future}

We begin this section with the question that got us started with this project.

\begin{question}
Let $P$ and $\hat P$ be two complementary components of the development of a convex polyhedron such that the polygonal cap of $P$ under the cap construction algorithm is $\hat P$.
\begin{enumerate}
    \item Determine the folding lines interior of $P$ and $\hat P$ that realizes the convex polyhedron.
    \item Furthermore, classify $n$-sided polygons ($n \geq 5$) which result in congruent or similar polyhedrons.
\end{enumerate}
\end{question}

This would be a natural continuation on the work of Han~\cite{Han2018poster}, where she studied and classified triangles and quadrilaterals.
It turned out to be a significantly more involved question for pentagons, which also required us to have an improved understanding of the closed cap condition.

As indicated in Section~\ref{sec:introduction}, the closed cap condition cannot tell whether the resulting polygonal cap curve intersects with itself.
Thus the equivalent condition in Theorem~\ref{thm:linearRelation} is a \emph{local} condition at best---if $P$ satisfies the closed cap condition, then there is an $(n-1)$-complex-dimensional neighborhood of polygons that also satisfies the closed cap condition.

\begin{question}
\label{q:self-intersect}
Consider the subspace $V$ of $\bbC^n$ defined by $\ds \sum_{k=1}^n \omega^k v_k = 0$.
\begin{enumerate}
    \item Describe the subset of $V$ that represents all valid $n$-sided polygons, that is, simple closed polygonal curves of $n$ sides.
    \item Describe the subset of $V$ that represents all $n$-sided polygons with simple closed polygonal cap curves.
\end{enumerate}
\end{question}

The next question was proposed in~\cite{DeMarco2017}, which is an interesting take on the converse of Alexandrov's theorem.

\begin{question}
\label{q:alternativeCurvature}
Let $P$ be an $n$-sided polygon, with side lengths $\ell_1, \dots, \ell_n$ and internal angle~$\theta_k$ at each of its vertices~$v_k$.
\begin{enumerate}
    \item Give an explicit description of all distributions of angular defect $\kappa(v_k)$ at the vertices $v_k$ so that $P$ satisfies the closed cap condition with respect to~$\kappa$.
    \item Furthermore, provide conditions under which the resulting convex polyhedron is planar.
\end{enumerate}
\end{question}

Finally, we present a question that requires some advanced background knowledge, which puts this study into the broader picture of research in harmonic caps.

Given a Jordan domain $P$ and the associated harmonic measure $\mu(\bbR^2 \backslash P, \infty)$, we construct a polygonal cap in the following way.
Let $\Phi \colon \bbC \backslash \overline{\bbD} \to \bbC \backslash \overline{P}$ be a conformal isomorphism that fixes $\infty$.
Carath\'eodory's theorem~\cite{GarnettMarshall2005} guarantees that $\Phi$ extends continuously onto $\partial P$.
This allows
\begin{enumerate}
    \item the vertices $v_{n,k} = \Phi(e^{2\pi k/n})$ to approximate $\partial P$, and
    \item the atomic measure $\kappa_n = \frac{1}{n}\sum_{k=1}^n \delta(v_{n,k})$
    to approximate the harmonic measure $\mu$,
\end{enumerate}
when $n$ is sufficiently large.
If the vertices $v_{n,k}$ are the inputs of the cap construction algorithm, the resulting polygonal cap curves $\hat P_n$ appear to converge to a certain shape as $n$ increases.
This process was used in~\cite{DeMarco2017,Han2018poster,Weng2020thesis} to approximate the harmonic cap $\hat P_\mu$, while approximations to the conformal isomorphism $\Phi$ were obtained using various methods including Leja points~\cite{Leja}, the Zipper algorithm~\cite{Marshall2007Zipper}, and Schwartz-Christoffel mappings~\cite{Driscoll1996SCToolbox}.
However, polygons formed by $v_{n,k}$ often miss the closed cap condition by a small margin, and as $n$ increases, the gap between $\hat v_{n,1}$ and $\hat v_{n,n+1}$ gets smaller and smaller.
The question we pose here seeks a rigorous justification of approximating harmonic caps with this process.

\begin{question}
Prove that $\hat P_n \to \hat P_\mu$ in some appropriate sense.
\end{question}

\bibliographystyle{acm}
\bibliography{bibliography}

\end{document}